\documentclass[a4paper, 11pt]{article}

\usepackage{amsthm}
\usepackage{amsmath}
\usepackage{amssymb}
\usepackage{graphicx}
\usepackage[round]{natbib}
\usepackage{bbm}
\usepackage{textcomp}
\usepackage{url}
\usepackage[margin=2.5cm]{geometry}
\usepackage{pdflscape}
\usepackage{subcaption}
\usepackage{booktabs}
\usepackage[utf8]{inputenc}
\usepackage{csquotes}

\usepackage{authblk}

\usepackage{todonotes}

\usepackage{algorithm}
\usepackage{algpseudocode}

\theoremstyle{plain}  
\newtheorem{thm}{Theorem}[section]

\theoremstyle{definition}

\theoremstyle{remark}

\newcommand{\F}{\mathcal{F}}

\newcommand{\one}{\mathbbm{1}}

\newcommand\blfootnote[1]{%
	\begingroup
	\renewcommand\thefootnote{}\footnote{#1}%
	\addtocounter{footnote}{-1}%
	\endgroup
}

\usepackage{hyperref}
\hypersetup{
	colorlinks   = true, 
	urlcolor     = blue, 
	linkcolor    = blue, 
	citecolor   = blue 
}

\begin{document}
	
	\title{A safe Hosmer-Lemeshow test}
	
	\author[1]{Alexander Henzi}
	\author[2]{Marius Puke}
	\author[3]{Timo Dimitriadis}
	\author[4]{Johanna Ziegel}

	\affil[1]{ETH Z\"{u}rich, Seminar for Statistics, Z\"{u}rich, Switzerland. 
		\href{mailto:alexander.henzi@stat.math.ethz.ch}{alexander.henzi@stat.math.ethz.ch}}
	
	\affil[2]{University of Hohenheim, Department of Economics and Computational Science Hub, Stuttgart, Germany. 
		\href{mailto:marius.puke@uni-hohenheim.de}{marius.puke@uni-hohenheim.de}}
	
	\affil[3]{Heidelberg University, Alfred-Weber-Institute of Economics and Heidelberg Institute for Theoretical Studies (HITS), Heidelberg, Germany. 
		\href{mailto:timo.dimitriadis@awi.uni-heidelberg.de}{timo.dimitriadis@awi.uni-heidelberg.de}}
	
	\affil[4]{University of Bern, Institute of Mathematical Statistics and Actuarial Science, Bern, Switzerland. 
		\href{mailto:johanna.ziegel@stat.unibe.ch}{johanna.ziegel@stat.unibe.ch}}
	
	\setcounter{Maxaffil}{0}
	\renewcommand\Affilfont{\itshape\small}

	\maketitle
	
	\begin{abstract}
		This article proposes an alternative to the Hosmer-Lemeshow (HL) test for evaluating the calibration of probability forecasts for binary events. The approach is based on e-values, a new tool for hypothesis testing. An e-value is a random variable with expected value less or equal to one under a null hypothesis. Large e-values give evidence against the null hypothesis, and the multiplicative inverse of an e-value is a p-value. 
		Our test uses online isotonic regression to estimate the calibration curve as a `betting strategy' against the null hypothesis. 
		We show that the test has power against essentially all alternatives, which makes it theoretically superior to the HL test and at the same time resolves the well-known instability problem of the latter.
		A simulation study shows that a feasible version of the proposed eHL test can detect slight miscalibrations in practically relevant sample sizes, but trades its universal validity and power guarantees against a reduced empirical power compared to the HL test in a classical simulation setup.
		We illustrate our test on recalibrated predictions for credit card defaults during the Taiwan credit card crisis, where the classical HL test delivers equivocal results.\blfootnote{The first two authors contributed equally to this work.} \\[0.8em]
	\end{abstract}

	\section{Introduction}
	
	Suppose that we have a sample of observations $(p_i,y_i)_{i=1}^n$ from $(P_i,Y_i)_{i=1}^n$ such that $(P_i,Y_i)$ has the same distribution as $(P,Y) \in [0,1]\times \{0,1\}$, $i=1,\dots,n$.
	The interpretation is that $P_i$ is a prediction for the probability that $Y_i = 1$. 
	The random variables are defined on some underlying probability space $(\Omega,\F)$ and  $\mathcal{P}$ denotes all probability measures on $(\Omega,\F)$. \citet{HosmerLemeshow1980} propose a test for the null hypothesis of perfect calibration 
	\[
	\mathcal{H}_{\text{HL},n} = \{\mathbb{P} \in \mathcal{P}\;|\; \mathbb{E}_{\mathbb{P}}(Y_i | P_i) = P_i\; \text{ $\mathbb{P}$-almost surely},\; i=1,\dots,n\}.
	\]

	The Hosmer-Lemeshow (henceforth HL) test is based on partitioning the interval $ [0,1] $ in $ g \in \mathbb{N} $ bins and counting the  observed numbers of events, $ o_{1g} $, and no event occurrences, $ o_{0g} $, in each bin. Based on that binning and counting procedure, the HL test statistic  to test for perfect calibration of the probability predictions is 
	\begin{align}
		\label{eq:HLTest}
		\widehat C = \sum_{k=1}^{g} \left[ \dfrac{(o_{1k} - \widehat e_{1k})^2}{\widehat e_{1k}} + \dfrac{(o_{0k} - \widehat e_{0k})^2}{\widehat e_{0k}} \right] ,
	\end{align} 
	where $ \widehat{e}_{1k} $ and $ \widehat{e}_{0k} $ are the expected event and no event occurrences in bin $ k $, respectively \citep{HosmerLemeshowSturdivant2013}.
	Under the null hypothesis, $\widehat C$ asymptotically follows a $\chi^2$-distribution with $g$ degrees of freedom given that the sample $(P_i,Y_i)_{i=1}^n$ was not used for model estimation (and $g-2$ degrees of freedom otherwise).
	
	Technically, the choice of the binning procedure is up the user of the HL test and is conventionally implemented via quantile based binning strategies with $g=10$, resulting in equally populated bins (decile-of-risk). Less commonly, the test is based on equidistantly spaced bins, where the unit interval (or the range of prediction values) is divided  into $g$ equidistant bins.
	While little attention is devoted to the binning procedure in practical applications, it implicitly determines the set of alternatives the test has power against \citep[Section 5]{Dimitriadis2022Honest}, such that the test result is often highly sensitive to the exact implementation of the binning; see e.g., \citet{Hosmeretal1997, Bertolini2000, Kuss2002} and our empirical application in Section \ref{sec:application}.
	Nevertheless, the HL test is still the literature's favorite for checking the calibration of binary prediction models and commonly used in current and highly influential medical and epidemiological studies; see amongst many others \citet{HLappl1, HLappl3, HLappl4}.
	

	In this article, we suggest a safe and stable HL test based on e-values (that we describe below) and isotonic regression \citep{Ayer1955, Brunk1965}.
	The test is henceforth called eHL test.
	\cite{Dimitriadise2016191118} recently propose the use of isotonic regression to resolve the closely related instability issue stemming from binning approaches in so-called reliability diagrams in forecast evaluation. 
	While feasible inference on the isotonic regression for classical testing procedures is hampered by complicated asymptotic distributions and an inconsistency of the bootstrap, the e-values adopted here prove to be an appealing alternative in this setting.
	Based on online isotonic regression studied by \citet{Kotlowski2017}, we show that (an ideal version of) our eHL test has power against essentially all deviations from calibration, which makes it theoretically superior to the classical HL test.
	
	E-values, where `e' abbreviates the word `expectation',  were proposed recently as an alternative to p-values in testing problems. In a nutshell, an e-value is a realization of a non-negative random variable whose expected value is at most one under a given null hypothesis. This already signals that an e-value itself allows for meaningful interpretations since an e-value greater than one provides evidence against the null hypothesis. 
	Additionally, the multiplicative inverse of an e-value is a conservative p-value by Markov's inequality. From a game-theoretic perspective, the e-value has a simple financial meaning in the sense that the e-value can be seen as the factor by which a skeptic multiplies her money when betting against the null hypothesis; see \citet{ShaferVovk2019,Shafer2021}. 
	
	An important advantage of e-values over p-values is their uncomplicated behavior in combinations: the arithmetic average of e-values also is an e-value, likewise the product of independent or successive e-values; see \citet{Shafer2021,GrunwaldHeideETAL2019,WangRamdas2020}. In practice, this appeals because more evidence can be added later, i.e.~evidence across studies can easily be combined. 
	
	The proposed eHL test offers a safe alternative to a fragile state-of-the-art approach by avoiding ad-hoc choices and software instabilities.
	It can be regarded as an application of the Universal Inference approach of \citet{Wasserman2020}. While this method allows to construct valid tests under only weak assumptions, it has been observed that this validity often comes at the price of a diminished power \citep{Strieder2022, Tse2022}. In Section \ref{sec:construction_eHL}, we show that an ideal -- but computationally infeasible -- variant of the eHL test does have guaranteed power to detect essentially all violations of calibration. Our proof relies on connections between the proposed e-value and the regret in random permutation online isotonic regression, which is studied by \citet{Kotlowski2017}. It has been observed that power guarantees for anytime-valid tests can be obtained by means of regret bounds of online prediction methods, see for example the discussion in \citet{Casgrain2022}. Our result demonstrates that such a connection also exists in the batch case of $e$-values for a fixed sample size $n$ due to connections with the online random permutation setting.

	In Section \ref{sec:simulation}, we compare a feasible version of the eHL test to the classical HL test in a simulation study.
	As expected, we find that the eHL test has conservative rejection rates under the null hypothesis and quickly develops power under model mispecification.
	While its empirical test power is lower than the one of the classical HL test, we do not consider this to be problematic as HL tests are often carried out in cases of vast data sets and are even criticized as being ``too powerful'' in that they reject essentially all, even acceptably well calibrated models \citep{Paul2013, Nattino2020}.
	See \cite{Dimitriadis2022Honest} for an alternative solution to this problem based on confidence bands.

	We apply the eHL test in Section \ref{sec:application} to predictions of a logistic regression model for the binary event of credit card defaults in Taiwan in 2005, where over-issuing of credit cards lead to many default payments and a subsequent credit card crisis \citep{YehLien2009}.
	The eHL test provides clear evidence against calibration of the logistic model predictions, and further illustrates that recalibration methods work well.
	In contrast, the classical HL test based on different natural binning choices delivers equivocal results with p-values ranging from 0 to 0.91 for a single prediction method, implying that a researcher could have cherry-picked the binning specification and hence the test result to her will.
	
	\section{Construction of HL e-values \label{sec:construction_eHL}}
	
	\subsection{Preliminaries}
	
	An e-variable for $\mathcal{H}_{\text{HL},n}$ is a non-negative random variable $E$ (that is allowed to take the value $+\infty$) such that $\mathbb{E}_{\mathbb{P}}(E) \le 1$ for all $\mathbb{P} \in \mathcal{P}$. An e-value is a realization of an e-variable. An e-variable $E$ always yields a valid p-variable $1/E$ (a p-value is a realized p-variable) by Markov's inequality, since
	\begin{equation}\label{eq:Markov}
		\mathbb{P}\Big(\frac{1}{E} \le \alpha\Big) = \mathbb{P}\Big(E \ge \frac{1}{\alpha}\Big) \le \alpha \mathbb{E}_{\mathbb{P}}(E) \le \alpha, \quad \text{for all $\mathbb{P} \in \mathcal{H}_{\text{HL},n}$}.
	\end{equation}
	We reject the null hypothesis $\mathcal{H}_{\text{HL},n}$ if we observe a large value of $E$. If we want to ensure a classical p-guarantee then we have to determine the rejection region for a given $\alpha$ by \eqref{eq:Markov}. \citet{VovkWang2021} show that this is essentially the only way to transform an e-variable into a p-variable.
	We say that an e-variable has the alternative hypothesis $\mathcal{H'} \subset \mathcal{P}$ if $\mathbb{E}_{\mathbb{Q}}(E) > 1$ for all $\mathbb{Q} \in \mathcal{H}'$.

	\subsection{Sample size one}\label{sec:HL1lambda}
	We first construct e-variables for the sample size one Hosmer-Lemeshow null hypothesis
	\[
	\mathcal{H}_{\text{HL},1} = \{\mathbb{P} \in \mathcal{P}\;|\; \mathbb{E}_{\mathbb{P}}(Y | P) = P \}.
	\]
	In the special case here, e-variables are likelihood ratios conditional on $P$. Indeed, if $q \in [0,1]$, an e-variable for $\mathcal{H}_{\text{HL},1}$ is given by
	\[
	E_q(P,Y) = \frac{q^Y(1-q)^{1-Y}}{P^Y(1-P)^{1-Y}} = \begin{cases}q/P, & \text{if $Y = 1$,}\\ (1-q)/(1-P), &\text{if $Y=0$.} \end{cases}
	\]
	The variable $E_q(P,Y)$ is clearly non-negative, and for $\mathbb{P} \in \mathcal{H}_{\text{HL},1}$,
	\begin{align*}
		\mathbb{E}_{\mathbb{P}}(E_q(P,Y)) & = \mathbb{E}_{\mathbb{P}}\left(\mathbb{E}_{\mathbb{P}}(Y \mid P)\frac{q}{P}  + \mathbb{E}_{\mathbb{P}}(1-Y \mid P)\frac{1-q}{1-P} \right) \\
		& = \mathbb{E}_{\mathbb{P}}\left(P\frac{q}{P} + (1-P) \frac{1-q}{1-P}\right) = 1.
	\end{align*}
	To find alternative hypotheses for the e-variable $E_q$, let $\bar{\pi} = \mathbb{E}_{\mathbb{Q}}(Y \mid P)$. Then, 
	\begin{align*}
		\mathbb{E}_{\mathbb{Q}}(E_q(P, Y) \mid P) & = \bar{\pi} \frac{q}{P} + (1 - \bar{\pi})\frac{1-q}{1-P} 
	\end{align*}
	is strictly larger one if and only if, $\bar{\pi} > P$ and $q > P$, or, $\bar{\pi} < P$ and $q < P$, i.e., if $\pi$ and $q$ are to the same side of $P$.
	This shows that if $q < P$, $E_q$ has the alternative
	\begin{equation}\label{eq:alternative_q<P}
		\mathcal{H}' = \{\mathbb{Q} \in \mathcal{P} \mid \mathbb{E}_{\mathbb{Q}}(Y\mid P) < P\},
	\end{equation}
	and if $q > P$, $E_q$ has the alternative
	\begin{equation}\label{eq:alternative_q>P}
		\mathcal{H}' = \{\mathbb{Q} \in \mathcal{P} \mid \mathbb{E}_{\mathbb{Q}}(Y\mid P) > P\}.
	\end{equation}
	
	
	It is possible to show that basically any e-variable for $\mathcal{H}_{\text{HL},1}$ is of the form $E = E_q(P, Y)$ for some $q$ (depending on $P$) but this requires some more arguments; it follows by the construction in \citet{HenziZiegel2021}, see also \citet{Waudby-SmithRamdas2020}. 
	The connection of $E_q(P,Y)$ to the e-variables in \citet{HenziZiegel2021} of type $E = 1 + \lambda D$ with $D \geq -1$ such that $\mathbb{E}_{\mathbb{P}}(D) = 0$ for $\mathbb{P} \in \mathcal{H}_{\text{HL},1}$, follows from the fact that $\lambda$ in this representation can be bijectively mapped to $q$. In this context, 
	\begin{equation}\label{eq:Elambda}
		E = 1 + \lambda (P - Y)
	\end{equation} is an  e-variable for $\mathcal{H}_{\text{HL},1}$ for any $\lambda$ that is $\sigma(P)$-measurable with $-(1/P) \le \lambda \le 1/(1-P)$. If $P = 1$, there is no restriction on $\lambda$ from above, and analogously if $P = 0$, there is no restriction from below. By choosing $\lambda = (P-q)/(P(1-P))$, we obtain that $E = E_q(P,Y)$.
	
	Clearly, the e-variable $E_q(P,Y)$ may take the value infinity if either $P = 0$ and $Y = 1$ or $P = 1$ and $Y = 0$ occurs; a single observation $Y = 1$ or $Y = 0$ is sufficient to reject the hypothesis of calibration with certainty if the predicted probabilities are in $\{0,1\}$. For the remainder of the theoretical part of this paper, we will always make the assumption $\mathbb{P}(P \in \{0,1\}) = 0$ to exclude these special but uninteresting cases.
	
	%

	\subsection{Combining e-values in the iid case \label{sec:comb_evals}}
	
	We assume now that $(P_i,Y_i)_{i=1}^n$ are independent and identically distributed (iid).
	For testing $\mathcal{H}_{\text{HL},n}$, we suggest the e-variable
	\begin{equation}\label{eq:EHLn}
		E_{\text{HL},n}^{\text{id}} = \prod_{i=1}^n E_{q_i}(P_i,Y_i),
	\end{equation}
	where $q_i$ is $\sigma(P_1,\dots,P_i,Y_1,\dots,Y_{i-1})$-measurable. For $\mathbb{P}\in \mathcal{H}_{\text{HL},n}$, we have
	\begin{align*}
		\mathbb{E}_{\mathbb{P}}E_{\text{HL},n}^{\text{id}} &= \mathbb{E}_{\mathbb{P}}\Big(\mathbb{E}_{\mathbb{P}}\Big(\prod_{i=1}^n E_{q_i}(P_i,Y_i)|P_1,\dots,P_n,Y_1,\dots,Y_{n-1}\Big)\Big)\\
		&= \mathbb{E}_{\mathbb{P}}\Big(\prod_{i=1}^{n-1} E_{q_i}(P_i,Y_i)\mathbb{E}_{\mathbb{P}}\Big(E_{q_n}(P_n,Y_n)|P_1,\dots,P_n,Y_1,\dots,Y_{n-1}\Big)\Big)\\
		&= \mathbb{E}_{\mathbb{P}}\Big(\prod_{i=1}^{n-1} E_{q_i}(P_i,Y_i)\Big(1 + \frac{P_n - q_n}{P_n(1-P_n)} \mathbb{E}_{\mathbb{P}}(P_n - Y_n |P_1,\dots,P_n,Y_1,\dots,Y_{n-1})\Big)\Big)\\
		&= \mathbb{E}_{\mathbb{P}}\Big(\prod_{i=1}^{n-1} E_{q_i}(P_i,Y_i)\Big(1 + \frac{P_n - q_n}{P_n(1-P_n)}  \mathbb{E}_{\mathbb{P}}(P_n - Y_n |P_n)\Big)\Big)\\
		& = \mathbb{E}_{\mathbb{P}}\Big(\prod_{i=1}^{n-1} E_{q_i}(P_i,Y_i)\Big) = \mathbb{E}_{\mathbb{P}}E_{\text{HL},n-1}^{\text{id}} = \dots =  1,
	\end{align*}
	where we used the equivalent representation of $E_q(P,Y)$ in \eqref{eq:Elambda}.
	In particular, from the above derivation it is easy to see that $(E_{\text{HL},n}^{\text{id}})_{n \in \mathbb{N}}$ is a test martingale. 
	
	The e-variable $E_{\text{HL},n}^{\text{id}}$ depends on the ordering of $(P_i,Y_i)_{i=1}^n$ through the choice of $q_i$. Let $S_n$ denote all permutations of $\{1,\dots,n\}$, and for $\sigma \in S_n$ define $E_{\text{HL},n}^{\sigma}$ as $E_{\text{HL},n}^{\text{id}}$ for the random variables $(P_{\sigma(i)},Y_{\sigma(i)})_{i=1}^n$ instead of $(P_i,Y_i)_{i=1}^n$. Generally, 
	\[
	\sup_{\sigma \in S_n} E_{\text{HL},n}^{\sigma}
	\]
	is not an e-variable for $\mathcal{H}_{\text{HL},n}$, so one would guess that there are opportunities to fish for (spurious) significance by choosing some specific ordering of a sample of observations $(p_i,y_i)_{i=1}^n$. If there is a natural ordering of the observations such as a time stamp then the problem usually does not occur in applications since a different ordering of the observations is hard to justify. Indeed, when the observations are sequential (and possibly dependent), the e-variable defined at \eqref{eq:EHLn} is also an e-variable for the hypothesis 
	\[
	\mathcal{H}_{\text{HL},n,seq} = \{\mathbb{P} \in \mathcal{P}\;|\; \mathbb{E}_{\mathbb{P}}(Y_i | P_1,\dots,P_{i},Y_1,\dots,Y_{i-1}) = P_i\; \text{$\mathbb{P}$-almost surely},\; i=1,\dots,n\}.
	\]
	Contrary to classical theory, the sequential case is easier to treat than the iid case and has been the focus of many works employing e-values including for example \citet{Waudby-SmithRamdas2020,HenziZiegel2021}.
	
	Coming back to our situation with iid data, an alternative to \eqref{eq:EHLn} could be
	\[
	E_{\text{HL},n,\text{sym}} = \frac{1}{n! } \sum_{\sigma \in S_n} E_{\text{HL},n}^{\sigma}.
	\]
	This strategy is essentially the merging technique for independent e-values in Section 4 of \citet{VovkWang2021}, and the object of interest in this article.

	\subsection{An ideal test with power guarantees}
	\label{sec:IdealTest}
	
	The statistic $E_{\text{HL},n,\text{sym}}$ is an e-variable solely under the requirement that for $i = 1, \dots, n$ and all permutations $\sigma$, the probabilities $q_{\sigma(i)}$ in $E_{\text{HL},n}^{\sigma}$ are a measurable function of $(P_{\sigma(j)}, Y_{\sigma(j)})$, $j = 1, \dots, i - 1$, and of $P_{\sigma(i)}$. In the following, we write
	\[
	q_{\sigma, \sigma(i)} = f_i(P_{\sigma(1)}, \dots, P_{\sigma(i)}, Y_{\sigma(1)}, \dots, Y_{\sigma(i-1)}),
	\]
	using the same algorithm $f_i$ for constructing $q_{\sigma, \sigma(i)}$ based on $P_{\sigma(1)}, \dots, P_{\sigma(i)}, Y_{\sigma(1)}, \dots, Y_{\sigma(i-1)}$ for all permutations $\sigma$. The challenge is then how to choose the functions $f_1, \dots, f_n$ such that the test has power. As argued by \citet{GrunwaldHeideETAL2019,Shafer2021}, a suitable measure of power for e-values is the growth rate $\mathbb{E}_{\mathbb{Q}}[\log(E)]$ under an alternative distribution $\mathbb{Q}$, so that ideally, $E$ grows exponentially fast in the sample size if the null hypothesis is violated. 
	
	Our algorithm for choosing $f_1, \dots, f_n$ is inspired by permutation online isotonic regression, studied extensively by \citet{Kotlowski2017}. In machine learning applications, isotonic regresison is an established method for the recalibration of binary classifiers; see e.g.\ \citet{ZadroznyElkan2002}  or \citet{Flach2012}. Recently, \citet{Dimitriadise2016191118} related the isotonic regression approach to reliability diagrams, which are a key diagnostic tool in evaluating probability forecast for binary events, especially in meteorology. Our results demonstrate that isotonic regression is also suitable for constructing universal tests of calibration.
	
	To introduce the algorithm for constructing our e-variable, let $p_1, \dots, p_i \in [0,1]$ be probability predictions and $y_1, \dots, y_i \in \{0,1\}$ be the corresponding outcomes. Then the isotonic regression of $y_1, \dots, y_i$ on $p_1, \dots, p_i$ can be described as the maximizer of
	\begin{equation} \label{eq:isoreg}
		\hat{R}_n(g_1, \dots, g_i) = \hat{R}(g_1, \dots, g_i; \, p_1, \dots, p_i, y_1, \dots, y_i) = \sum_{j=1}^i \log \, \left( \left(\frac{g_j}{p_j}\right)^{y_j}\left(\frac{1-g_j}{1-p_j}\right)^{1-y_j} \right),
	\end{equation}
	over all $g_1, \dots, g_i$ such that $g_k \leq g_l$ if $p_k \leq p_l$. 
	Notice that the quantity in \eqref{eq:isoreg} is simply a normalized version of the logarithmic score, and the maximizer does not depend on the fact that we normalize by $p_j^{y_j}(1-p_j)^{1-y_j}$. 
	Moreover notice that up to rescaling by $1/i$, this criterion also equals the sample version of $\mathbb{E}_{\mathbb{Q}}[\log(E)]$ when the e-variable $E$ is the likelihood ratio between the probabilities $g_j$ and $p_j$. A unique maximizer exists --- unique since we exclude the cases $p_j = 0$ and $y_j = 1$ or $p_j = 1$ and $y_j = 0$ for some $j$ --- and can be computed efficiently with the PAV-Algorithm \citep{Ayer1955}. This estimator only defines a recalibrated version of $p_1, \dots, p_i$, and a method is required to define the regression at a $p_{i+1} \in [0,1]$ not contained in the sample. To obtain out-of-sample predictions with small regret in terms of log-loss, we rely on a strategy originally proposed by \citet{Vovk2015} and applied by \citet{Kotlowski2017} to derive regret bounds for isotonic regression in an online setting. The out-of-sample value at $p_{j+1}$ is defined as follows, 
	\begin{equation}  \label{eq:isoreg_oos}
		f_{i+1}(p_1, \dots, p_i, p_{i+1}, y_1, \dots, y_i) = \frac{g_{i+1,1}}{g_{i+1,1}+1-g_{i+1,0}},
	\end{equation}
	where $g_{i+1,1}$ and $g_{i+1,0}$ are the $(i+1)$-th component the isotonic regression of $p_i, \dots, p_i, p_{i+1}$ with observations $y_1, \dots, y_i, 1$ or $y_1, \dots, y_i, 0$, respectively. 
	That is, to define the isotonic regression at the unseen $p_{i+1}$, we fit two isotonic regression in which we include $p_{i+1}$ in the sample with artificial observations of $1$ and of $0$ respectively, and take the ratio \eqref{eq:isoreg_oos} as recalibrated probability. The definition \eqref{eq:isoreg_oos} is extended to the case $i = 0$ by setting $g_{1,1} = g_{1,0} = 0.5$. The workflow to construct $E_{\text{HL},n,\text{sym}}$ is then described in Algorithm \ref{alg:esym}.
	
	\begin{algorithm}
		\caption{Construction of $E_{\text{HL},n,\text{sym}}$}\label{alg:esym}
		\begin{algorithmic}[1]
			\State $E_{\text{HL},n,\text{sym}} \gets 0$
			\For{all permutations $\sigma$ of $\{1, \dots, n\}$}
			\State $E_{\text{HL},n}^{\sigma} \gets (0.5/P_{\sigma(1)})^{Y_{\sigma(1)}}(0.5/(1-P_{\sigma(1)}))^{1-Y_{\sigma(1)}}$
			\For{$i=1,\dots,n-1$}
			\State $q_{\sigma, \sigma(i+1)} \gets f_{i}(P_{\sigma(1)}, \dots, P_{\sigma(i+1)}, Y_{\sigma(1)}, \dots, Y_{\sigma(i)})$ as defined in \eqref{eq:isoreg_oos}
			\State $E_{\text{HL},n}^{\sigma} \gets E_{\text{HL},n}^{\sigma} \cdot (q_{\sigma, \sigma, \sigma(i+1)}/P_{\sigma(i+1)})^{Y_{\sigma(i+1)}}((1-q_{\sigma, \sigma(i+1)})/(1-P_{\sigma(i+1)}))^{1-Y_{\sigma(i+1)}}$
			\EndFor
			\State $E_{\text{HL},n,\text{sym}} \gets E_{\text{HL},n,\text{sym}} + E_{\text{HL},n}^{\sigma}/n!$
			\EndFor
			\State \textbf{return} $E_{\text{HL},n,\text{sym}}$
		\end{algorithmic}
	\end{algorithm}
	
	To state a result about the power of $E_{\text{HL},n,\text{sym}}$, we need a population version of the isotonic regression estimator. For a function $\pi \colon [0,1] \rightarrow [0,1]$, let
	\[
	R_{\mathbb{Q}}(\pi) = \mathbb{E}_{\mathbb{Q}}\left[\log \left( (\pi(P)/P)^Y((1-\pi(P))/(1-P))^{1-Y}\right)\right]
	\]
	if this expectation exists. Let $\mathcal{F}_{\uparrow, [0,1]}$ be the set of nondecreasing functions $\pi \colon [0,1] \rightarrow [0,1]$. If $\mathbb{Q}$ is the empirical distribution of $(P_1, Y_1), \dots, (P_n, Y_n)$, then it is easy to see that $R_{\mathbb{Q}}$ coincides with the target function $\hat{R}$ of the usual isotonic regression in finite samples. With these definitions, we can state the following result about the power of $E_{\text{HL},n,\text{sym}}$.
	
	\begin{thm} \label{thm:power}
		Let $(P_1, Y_1), \dots, (P_n, Y_n), (P, Y)$ be iid with distribution $\mathbb{Q}$ such that
		\begin{equation} \label{eq:integrability}
			\mathbb{E}_{\mathbb{Q}}[\log(P)^2 + \log(1-P)^2] < \infty.
		\end{equation}
		Then,
		\begin{itemize}
			\item[(i)] there exists a $\mathbb{Q}$-almost-surely unique maximizer $\pi^* \in \mathcal{F}_{\uparrow, [0,1]}$ of $R_{\mathbb{Q}}$;
			\item[(ii)] for a version of $\pi^*$ from part (i), let
			\[
			D(\mathbb{Q}) = R_{\mathbb{Q}}(\pi^*) = \mathbb{E}_{\mathbb{Q}}[\log \, (\pi^*(P)/P)^Y((1-\pi^*(P))/(1-P))^{1-Y}];
			\]
			then $D(\mathbb{Q}) \geq 0$, with equality if and only if $\mathbb{Q} \in \mathcal{H}_{\text{HL},1}$;
			\item[(iii)] the e-value $E_{\text{HL},n,\text{sym}}$ constructed with Algorithm \ref{alg:esym} satisfies
			\[
			E_{\text{HL},n,\text{sym}} \geq \exp\left(\sum_{i=1}^{n}\log\, \left(\frac{\pi^*(P_i)}{P_i}\right)^{Y_i} \left(\frac{1-\pi^*(P_i)}{1-P_i}\right)^{1-Y_i} -C\sqrt{n\log(n)^2} \right)
			\]
			for an universal constant $C > 0$, and hence
			\[
			\mathbb{E}_{\mathbb{Q}}[\log(E_{\text{HL},n,\text{sym}})] \geq nD(\mathbb{Q}) - C\sqrt{n\log(n)^2}.
			\]
		\end{itemize}
	\end{thm}
	
	The integrability assumption \eqref{eq:integrability} is solely required to prove parts (i) and (ii) of the theorem, and the lower bound on $E_{\text{HL},n,\text{sym}}$ and the expectation of its logarithm in fact hold for any $\pi \in \mathcal{F}_{\uparrow, [0,1]}$. However, part (iii) only becomes useful in conjunction with (i) and (ii): the fact that $D(\mathbb{Q}) \geq 0$ with equality if and only if $\mathbb{Q} \in \mathcal{H}_{\text{HL},1}$ implies that the test has positive growth rate for all alternative distributions $\mathbb{Q}$ if $n$ is large enough. This is a surprising result, since it might seem that restricting our estimator of $\mathbb{E}_{\mathbb{Q}}[Y|P]$ to isotonic functions in $P$ implies some restriction on the class of alternatives against which the test has power --- which is not the case.
	
	Part (iii) of Theorem \ref{thm:power} only gives a diverging lower bound on the expected value of $\log(E_{\text{HL},n,\text{sym}})$. However, under assumption \eqref{eq:integrability} the average growth rate
	\[
	\frac{1}{n}\sum_{i=1}^{n}\log\, \left(\frac{\pi^*(P_i)}{P_i}\right)^{Y_i} \left(\frac{1-\pi^*(P_i)}{1-P_i}\right)^{1-Y_i}
	\]
	satisfies the strong law of large numbers, and since $D(\mathbb{Q}) > 0$ for $\mathbb{Q} \not\in \mathcal{H}_{\text{HL},1}$, this implies that $E_{\text{HL},n,\text{sym}} \rightarrow \infty$ almost surely as $n \rightarrow \infty$. In particular, for any Type-I error $\alpha$ and desired power $1-\beta$, there exists a sample size $N$ such that $\mathbb{Q}(E_{\text{HL},N,\text{sym}} \geq 1/\alpha) \geq 1-\beta$ for $N \geq n$.

	\begin{proof}[Proof of Theorem \ref{thm:power}]
		Part (i) is a consequence of the following facts. If $(\pi_n)_{n\in\mathbb{N}}$ is a sequence in $\mathcal{F}_{\uparrow, [0,1]}$ such that $\lim_{n\rightarrow\infty} R_{\mathbb{Q}}(\pi_n) = \sup_{\pi \in \mathcal{F}_{\uparrow, [0,1]}}R(\pi)$, then by Helly's selection theorem, there exists a subsequence $(\pi_{n_k})_{k\in\mathbb{N}}$ converging pointwise to some $\pi^* \in \mathcal{F}_{\uparrow, [0,1]}$. The function $\pi(P) \mapsto \log \left( (\pi(P)/P)^Y((1-\pi(P))/(1-P))^{1-Y}\right)$ inside the expectation in the definition of $R_{\mathbb{Q}}$ is strictly concave, and the set $\mathcal{F}_{\uparrow, [0,1]}$ is convex. Hence $R_{\mathbb{Q}}(\pi^*) \geq R_{\mathbb{Q}}(\pi)$ for all $\pi \in \mathcal{F}_{\uparrow, [0,1]}$, and equality holds if and only if $\pi = \pi^*$ $\mathbb{Q}$-almost-surely, provided that $R_{\mathbb{Q}}(\pi^*)$ is finite, which is shown below.
		
		The nonnegatitivy in part (ii) holds because $\mathcal{F}_{\uparrow, [0,1]}$ contains the identity function, and we only have to prove that $D(\mathbb{Q}) > 0$ if $\mathbb{Q} \not\in \mathcal{H}_{\text{HL},1}$. For this, it is sufficient to show that there exists some $\pi$ with $R_{\mathbb{Q}}(\pi) > 0$. We start with some results about the existence of certain expected values.
		Since $Y|P$ is Bernoulli with expectation $\bar{\pi}(P)$, we have
		\begin{align*}
			& \mathbb{E}_{\mathbb{Q}}[\log \, (\pi^*(P)/P)^Y((1-\pi^*(P))/(1-P))^{1-Y} \mid P] \\
			& \quad = \bar{\pi}(P)\log(\pi^*(P)/P) + (1-\bar{\pi}(P))\log((1-\pi^*(P))/(1-P))
		\end{align*}
		and the nonnegativity of the Kullback-Leibler divergence implies
		\begin{align*}
			& \bar{\pi}(P)\log(\pi^*(P)/P) + (1-\bar{\pi}(P))\log((1-\pi^*(P))/(1-P)) \\
			& \quad \leq \bar{\pi}(P)\log(\bar{\pi}(P)/P) + (1-\bar{\pi}(P))\log((1-\bar{\pi}(P))/(1-P)),
		\end{align*}
		hence we obtain
		\begin{align}
			0 & \leq \mathbb{E}_{\mathbb{Q}}[\log \, (\pi^*(P)/P)^Y((1-\pi^*(P))/(1-P))^{1-Y}] \nonumber\\
			& \leq \mathbb{E}_{\mathbb{Q}}[\bar{\pi}(P)\log(\bar{\pi}(P)/P) + (1-\bar{\pi}(P))\log((1-\bar{\pi}(P))/(1-P))] \nonumber \\
			& \leq \mathbb{E}_{\mathbb{Q}}[|\log(P)| + |\log(1-P)|] \nonumber \\
			& \leq \sqrt{\mathbb{E}_{\mathbb{Q}}[\log(P)^2]} + \sqrt{\mathbb{E}_{\mathbb{Q}}[\log(1-P)^2]} < \infty. \label{eq:finite_integral}
		\end{align}
		Let now $\tilde{\pi}(P)$ be a version of the conditional expectation of $\bar{\pi}(P)$ with respect to the sigma lattice generated by $P$, which $\tilde{\pi}(P)$ satisfies the following properties:
		\begin{align}
			& \tilde{\pi} \text{ is increasing }; \label{eq:expectation_lattice_incr} \\
			& \mathbb{E}_{\mathbb{Q}}[(\bar{\pi}(P) - \tilde{\pi}(P))h(P)] \leq 0 \text{ for all increasing } h \text{ such that } \mathbb{E}_{\mathbb{Q}}[h(P)^2] < \infty; \label{eq:expectation_lattice_ineq} \\
			& \mathbb{E}_{\mathbb{Q}}[\bar{\pi}(P)\one_B(P)] = \mathbb{E}_{\mathbb{Q}}[\tilde{\pi}(P)\one_B(P)] \text{ for all } B \text{ in the } \sigma\text{-field generated by } \tilde{\pi}. \label{eq:expectation_lattice_eq}
		\end{align}
		Equation \eqref{eq:expectation_lattice_incr} holds by definition of the conditional expectation given a sigma lattice, and \eqref{eq:expectation_lattice_ineq} and \eqref{eq:expectation_lattice_eq} are by \citet[Equations (3.9) and (3.11)]{Brunk1965}. By \eqref{eq:expectation_lattice_eq}, we have $\bar{\pi}(P) = P$ almost surely if and only if $\tilde{\pi}(P) = P$ almost surely. By definition of $\pi^*(P)$, we know that
		\[
		\mathbb{E}_{\mathbb{Q}}[\log \, (\tilde{\pi}(P)/P)^Y((1-\tilde{\pi}(P))/(1-P))^{1-Y}] \leq \mathbb{E}_{\mathbb{Q}}[\log \, (\pi^*(P)/P)^Y((1-\pi^*(P))/(1-P))^{1-Y}].
		\]
		The goal is now to prove
		\[
		\mathbb{E}_{\mathbb{Q}}\left[\log \left( (\tilde{\pi}(P)/P)^Y((1-\tilde{\pi}(P))/(1-P))^{1-Y}\right)\right] \geq 0,
		\]
		with equality if and only if $\tilde{\pi}(P) = P$ $\mathbb{Q}$-almost-surely. Notice that $\tilde{\pi}$ is only defined on the support of $P$, but one can assume without loss of generality that it is defined on the whole interval $[0,1]$ by right-continuous constant extrapolation in parts where it is not defined. Since $\tilde{\pi}$ is increasing, there exist at most countably many disjoint intervals $A_i \subseteq [0,1]$, $i \in \mathcal{I}$, on which $\tilde{\pi}$ is constant with some value $c_i \in [0,1]$. Furthermore, there are at most countably many disjoint intervals $B_j$, $j \in \mathcal{J}$, whose union equals $[0,1] \setminus \bigcup_{i \in \mathcal{I}} A_i$. We now make a few case distinctions.
		
		Fix $i$ with $c_i > 0$ and assume that $A_i = [a_i, b_i]$; the following arguments can be easily modified for the case that $A_i$ is (half-)open. Define the function
		\[
		h_i(x) = \begin{cases}
			\log(1/a_i) + 1, \ & \text{ if }x < a_i, \\
			\log(1/x), & \text{ if } x \in [a_i, b_i], \\
			-1, &\text{ if } x > b_i.
		\end{cases}
		\]
		Then $h_i(P)$ is square integrable due to \eqref{eq:finite_integral}, $h_i$ is decreasing, and constant outside of $[a_i, b_i]$, so that
		\begin{align*}
			0 \overset{\eqref{eq:expectation_lattice_ineq}}{\leq} \mathbb{E}_{\mathbb{Q}}[(\bar{\pi}(P) - \tilde{\pi}(P))h_i(P)] 
			& \overset{\eqref{eq:expectation_lattice_eq}}{=} \mathbb{E}_{\mathbb{Q}}[(\bar{\pi}(P) - \tilde{\pi}(P))\log(1/P)\one_{A_i}(P)] \\
			& \overset{\eqref{eq:expectation_lattice_eq}}{=} \mathbb{E}_{\mathbb{Q}}[(\bar{\pi}(P) - \tilde{\pi}(P))\log(c_i/P)\one_{A_i}(P)] \\
			& \, \,= \mathbb{E}_{\mathbb{Q}}[(\bar{\pi}(P) - \tilde{\pi}(P))\log(\tilde{\pi}(P)/P)\one_{A_i}(P)],
		\end{align*}
		where the last step is due to the fact that $\tilde{\pi}(P) = c_i$ for $P \in A_i$. Hence
		\[
		\mathbb{E}_{\mathbb{Q}}[\bar{\pi}(P)\log(\tilde{\pi}(P)/P)\one_{A_i}(P)] \geq 	\mathbb{E}_{\mathbb{Q}}[\tilde{\pi}(P)\log(\tilde{\pi}(P)/P)\one_{A_i}(P)]. 
		\]
		If $c_i = 0$, the above inequality is still true because \eqref{eq:expectation_lattice_eq} implies that $\tilde{\pi}(P) = \bar{\pi}(P) = 0$ for $P \in A_i$ in that case, and we define $0\log(0) := 0$.
		
		Similarly, for $c_i < 1$ we define
		\[
		h_i(x) = \begin{cases}
			\log(1/(1-b_i)) + 1, \ & \text{ if } x > b_i, \\
			\log(1/(1-x)), & \text{ if } x \in [a_i, b_i], \\
			-1, &\text{ if } x < a_i,
		\end{cases}
		\]
		which is square integrable and increasing. As before,
		\begin{align*}
			0 \overset{\eqref{eq:expectation_lattice_ineq}}{\leq} \mathbb{E}_{\mathbb{Q}}[(\tilde{\pi}(P) - \bar{\pi}(P))h_i(P)]
			& \overset{\eqref{eq:expectation_lattice_eq}}{=} \mathbb{E}_{\mathbb{Q}}[(\tilde{\pi}(P) - \bar{\pi}(P))\log(1/(1-P))\one_{A_i}(P)] \\
			& \overset{\eqref{eq:expectation_lattice_eq}}{=} \mathbb{E}_{\mathbb{Q}}[(\tilde{\pi}(P) - \bar{\pi}(P))\log((1-c_i)/(1-P))\one_{A_i}(P)] \\
			& \, \,= \mathbb{E}_{\mathbb{Q}}[(1-\bar{\pi}(P) - (1-\tilde{\pi}(P)))\log((1-\tilde{\pi}(P))/(1-P))\one_{A_i}(P)],
		\end{align*}
		so we obtain
		\[
		\mathbb{E}_{\mathbb{Q}}[(1-\bar{\pi}(P))\log((1-\tilde{\pi}(P))/(1-P))\one_{A_i}(P)] \geq 	\mathbb{E}_{\mathbb{Q}}[(1-\tilde{\pi}(P))\log((1-\tilde{\pi}(P))/(1-P))\one_{A_i}(P)], 
		\]
		which also holds if $\tilde{\pi}(P) = 1$ on $A_i$. Hence we have shown that
		\begin{align*}
			0 & \leq \mathbb{E}_{\mathbb{Q}}[\one_{A_i}(P)(\tilde{\pi}(P)\log(\tilde{\pi}(P)/P) + (1-\tilde{\pi}(P))\log((1-\tilde{\pi}(P))/(1-P)))] \\
			& \leq \mathbb{E}_{\mathbb{Q}}[\one_{A_i}(P)(\bar{\pi}(P)\log(\tilde{\pi}(P)/P) + (1-\bar{\pi}(P))\log((1-\tilde{\pi}(P))/(1-P)))] \\
			& = \mathbb{E}_{\mathbb{Q}}[\one_{A_i}(P)\log \, (\tilde{\pi}(P)/P)^Y((1-\tilde{\pi}(P))/(1-P))^{1-Y}],
		\end{align*}
		and equality holds if and only if $\tilde{\pi}(P) = P$ $\mathbb{Q}$-almost-surely on $A_i$, since the Kullback-Leibler divergence is non-negative.
		
		Consider now an interval $B_j$. Since $\tilde{\pi}$ is strictly increasing on $B_j$, the sigma field generated by $\tilde{\pi}$ contains all Borel sets which are subsets of $B_j$. Then \eqref{eq:expectation_lattice_eq} implies that $\tilde{\pi}(P) = \bar{\pi}(P)$ $\mathbb{Q}$-almost-surely on $B_j$, hence
		\begin{align*}
			& \mathbb{E}_{\mathbb{Q}}\left[\one_{B_j}(P)\left(\bar{\pi}(P)\log(\tilde{\pi}(P)/P) + (1-\bar{\pi}(P))\log((1-\tilde{\pi}(P))/(1-P))\right)\right] \\
			& \ =  \mathbb{E}_{\mathbb{Q}}\left[\one_{B_j}(P)\left(\tilde{\pi}(P)\log(\tilde{\pi}(P)/P) + (1-\tilde{\pi}(P))\log((1-\tilde{\pi}(P))/(1-P))\right)\right]\geq 0
		\end{align*}
		with equality if and only if $\tilde{\pi}(P) = P$ $\mathbb{Q}$-almost-surely on $B_j$.
		
		With the above derivations, we obtain that for any finite number of indices $i_1, \dots, i_n \in \mathcal{I}$, $j_1, \dots, j_n \in \mathcal{J}$ and 
		\[
		C_n = \left(\bigcup_{k=1}^n A_{i_k}\right) \cup \left(\bigcup_{l=1}^n B_{j_l}\right),
		\]
		the following inequalities hold,
		\begin{align}
			0 & \leq \mathbb{E}_{\mathbb{Q}}\left[\one_{C_n}(P)\left(\tilde{\pi}(P)\log(\tilde{\pi}(P)/P) + (1-\tilde{\pi}(P))\log((1-\tilde{\pi}(P))/(1-P))\right)\right] \label{eq:int_lower_bound}\\
			& \leq \mathbb{E}_{\mathbb{Q}}\left[\one_{C_n}(P)\left(\bar{\pi}(P)\log(\tilde{\pi}(P)/P) + (1-\bar{\pi}(P))\log((1-\tilde{\pi}(P))/(1-P))\right)\right] \label{eq:int_upper_bound}.
		\end{align}
		Since the integrand in \eqref{eq:int_lower_bound} is non-negative and the integrand in \eqref{eq:int_upper_bound} dominated pointwise by
		\[
		M(P) = \bar{\pi}(P)\log(\bar{\pi}(P)/P) + (1-\bar{\pi}(P))\log((1-\bar{\pi}(P))/(1-P))
		\]
		with $\mathbb{E}_{\mathbb{Q}}[M(P)] < \infty$, we can choose index sequences such that $\bigcup_{n=1}^N C_n$ increases to $[0,1]$, and apply Fatou's Lemma and the dominated convergence theorem to obtain
		\begin{align}
			0 & \leq \mathbb{E}_{\mathbb{Q}}\left[\tilde{\pi}(P)\log(\tilde{\pi}(P)/P) + (1-\tilde{\pi}(P))\log((1-\tilde{\pi}(P))/(1-P))\right] \label{eq:int_lower_bound2}\\
			& \leq \mathbb{E}_{\mathbb{Q}}\left[\bar{\pi}(P)\log(\tilde{\pi}(P)/P) + (1-\bar{\pi}(P))\log((1-\tilde{\pi}(P))/(1-P))\right] \nonumber.
		\end{align}
		Equality in \eqref{eq:int_lower_bound2} holds if and only if $\tilde{\pi}(P) = P$ almost surely.
		
		For part (iii), the inequality of arithmetic and geometric mean implies that 
		\[
		E \geq \left(\prod_{\sigma \in \mathcal{S}_n}\prod_{i=1}^{n} \frac{q_{\sigma,\sigma(i)}^{Y_{\sigma(i)}}(1-q_{\sigma,\sigma(i)})^{1-Y_{\sigma(i)}}}{P_{\sigma(i)}^{Y_{\sigma(i)}}(1-P_{\sigma(i)})^{1-Y_{\sigma(i)}}}\right)^{1/n!} = \exp\left(\frac{1}{n!}\sum_{\sigma \in \mathcal{S}_n} \sum_{i=1}^n \log\frac{q_{\sigma,\sigma(i)}^{Y_{\sigma(i)}}(1-q_{\sigma,\sigma(i)})^{1-Y_{\sigma(i)}}}{P_{\sigma(i)}^{Y_{\sigma(i)}}(1-P_{\sigma(i)})^{1-Y_{\sigma(i)}}}\right).
		\]
		The term inside the exponential can be written as
		\[
		L = \mathbb{E}_{\sigma}\left[\sum_{i=1}^n \log\frac{q_{\sigma,\sigma(i)}^{Y_{\sigma(i)}}(1-q_{\sigma,\sigma(i)})^{1-Y_{\sigma(i)}}}{P_{\sigma(i)}^{Y_{\sigma(i)}}(1-P_{\sigma(i)})^{1-Y_{\sigma(i)}}}\right],
		\]
		which is the negative of the entropic loss defined in Section 4.4 of \citet{Kotlowski2017}, and the expectation $\mathbb{E}_{\sigma}[\cdot]$ is with respect to the uniform distribution over all permutations $\sigma$ of $\{1,\dots,n\}$. It follows from Lemma 2.1, Theorem 4.3 and the proof of Theorem 4.1 of \citet{Kotlowski2017} that for all $K \in \mathbb{N}$,
		\[
		L - \sum_{i=1}^n \log\frac{\hat{\pi}_i^{Y_{i}}(1-\hat{\pi}_i)^{1-Y_{i}}}{P_{i}^{Y_{i}}(1-P_{\sigma(i)})^{1-Y_{i}}} \geq -\sum_{k=1}^n \left(\frac{2}{K}+\frac{4K}{k}\log(1+k)\right),
		\]
		where $\hat{\pi}_1, \dots, \hat{\pi}_n$ is the isotonic regression of $Y_1, \dots, Y_n$ on $P_1, \dots, P_n$, i.e.~the maximizer of 
		\[
		(g_1, \dots, g_n) \mapsto \hat{R}(g_1, \dots, g_n; \, P_1, \dots, P_n, Y_1, \dots, Y_n).
		\]
		as defined at \eqref{eq:isoreg}. The result now follows because
		\[
		\hat{R}(\hat{\pi}_1, \dots, \hat{\pi}_n; \, P_1, \dots, P_n, Y_1, \dots, Y_n) \geq \hat{R}(\pi^*(P_1), \dots, \pi^*(g_n); \, P_1, \dots, P_n, Y_1, \dots, Y_n)
		\]
		and $\sum_{k=1}^n (2/K + 4K\log(1+k)/k) = \mathcal{O}(\sqrt{n(\log(n))^2})$ for $K$ of order $\sqrt{n/(\log(n))^2}$.
	\end{proof}

	\subsection{A feasible version of the test} \label{sec:feasible}
	The ideal test described in Algorithm \ref{alg:esym} cannot be implemented for practically relevant $n$, as it requires to compute e-values over all $n!$ permutations of $\{1, \dots, n\}$. Even for a single permutation $\sigma$, the inner loop in Algorithm \ref{alg:esym} has computational complexity of $\mathcal{O}(n^2)$: it requires computing $2n$ isotonic regressions to generate out-of-sample predictions. We suggest to address these problems above by the simplified version in Algorithm \ref{alg:feasible}, which can be regarded as a version of the split likelihood ratio test by \citet{Wasserman2020}.
	
	\begin{algorithm}[tbh]
		\caption{Split LRT version of the e-value}\label{alg:feasible}
		\begin{algorithmic}[1]
			\State \textbf{Parameters:} split fraction $s \in (0,1)$, number of splits $B \in \mathbb{N}$.
			\State $E_{\text{HL},n} \gets 0$
			\For{$b = 1, \dots, B$}
			\State randomly select $\lfloor ns \rfloor$ pairs $(Y_{i}, P_{i})$, $j \in S_b = \{i_1, \dots, i_{\lfloor ns \rfloor}\}$, without replacement
			\State estimate the isotonic of regression of $(Y_{i}, P_{i})$, $i \in S_b$, by maximizing \eqref{eq:isoreg}
			\State generate predictions $q_i$ for $\mathbb{E}[Y|P_i]$, $i \in \{1, \dots, n\} \setminus S_b$, from the isotonic regression
			\State $E_{\text{HL},n} \leftarrow E_{\text{HL},n} + \prod_{i \in \{1, \dots, n\} \setminus S_b} (q_i/P_i)^{Y_i}((1-q_i)/(1-P_i))^{1-Y_i} /B$
			\EndFor
			\State \textbf{return} $E_{\text{HL},n}$
		\end{algorithmic}
	\end{algorithm}
	
	A delicate point in Algorithm \ref{alg:feasible} is Step 6, where one needs to generate out-of-sample predictions from the isotonic regression fit. Naive extrapolation approaches could lead to predicted probabilities $q_i \in \{0,1\}$ and hence an e-value of zero if either $q_i = 0$ and $Y_i = 1$ or $q_i = 1$ and $Y_i = 0$. 
	
	Let $p_1 < \dots < p_m$ denote the distinct values of $P_i$, $i \in S_b$, and $\hat{\pi}_1 \leq \dots \leq \hat{\pi}_m$ the corresponding values of the isotonic regression. A well known result about isotonic regression states that there exists a partition of $S_b$ into index sets $\mathcal{I}_1, \dots, \mathcal{I}_d$ such that $\hat{\pi}_j$ is the empirical mean of the $Y_i$ with indices in $\mathcal{I}_j$,
	\[
	\hat{\pi}_j = \frac{1}{\#\mathcal{I}_j}\sum_{i \in \mathcal{I}_j} Y_i.
	\]
	To remedy the problem of predictions in $\{0,1\}$, we propose to apply the smoothed Laplace predictor, equivalent to Jeffreys' prior in binomial proportion estimation,
	\[
	\check{\pi}_j = \frac{1}{\# \mathcal{I}_j + 1} \left(0.5 + \sum_{i \in \mathcal{I}_j}Y_i \right) \in (0,1).
	\]
	For out-of-sample predictions at $P_i \not\in \{p_1, \dots, p_m\}$, one can then apply linear interpolation
	\[
	q_i = \begin{cases} \displaystyle
		\frac{p_l - P_i}{p_l - p_k} \check{\pi}_k + \frac{P_i - p_k}{p_l - p_k} \check{\pi}_l, & \ \text{ if } P_i \in [p_k, p_l], \\[0.2em]
		\check{\pi}_1,  & \ \text{ if } P_i < p_1, \\[0.2em]
		\check{\pi}_m,  & \ \text{ if } P_i > p_m., \\
	\end{cases}
	\]
	where it is now guaranteed that $q_i \in (0,1)$.

	\section{Simulations}
	\label{sec:simulation}
	
	This section evaluates the empirical performance of the feasible version of the proposed test in Section \ref{sec:feasible} together with sensible values of the splitting fraction $s \in (0,1)$.
	We follow the simulation setup of \citet{Hosmeretal1997} with a quadratic misspecification in assessing HL-type tests, which is, if at all, just slightly modified in more recent contributions \citep{HosmerHjort2002, XieEtAl2008, Allison2014, CanaryEtAl2017, Nattino2020}.
	Replication material for the simulations and the application in Section \ref{sec:application} in the statistical software \texttt{R} is available under \href{https://github.com/marius-cp/eHL}{https://github.com/marius-cp/eHL}.
	
	For $i=1,\dots,2n$ with $n \in \{1024, 2048, 4096, 8192\}$, we simulate the iid covariate $X_i \stackrel{\text{iid}}{\sim} \text{U}(-3,3)$ and let the response variables $Y_i \sim \text{Bernoulli}(\pi_i)$ be independent, where the true conditional event probability $\bar{\pi}_i$ follows a logistic transformation of a quadratic model
	\begin{align} 
		\label{Eq:true}
		\bar{\pi}_i =  \bar{\pi} (X_i) = \mathbb{P} (Y_i=1 \mid X_i; \, \beta_0, \beta_1, \beta_2) = \dfrac{\exp( \beta_0 + \beta_1 X_i + \beta_2 X_i^2)}{1+\exp( \beta_0 + \beta_1 X_i + \beta_2 X_i^2)}.
	\end{align} 
	
	We split the simulated data into an estimation set and validation set, both of size $n$. 
	Based on the data in the estimation set, we estimate the parameters of a linear, and hence misspecified, logistic regression model by maximum likelihood and denote the parameter estimates by $\big(\widehat{\beta}_0, \widehat{\beta}_1\big)$.
	The probability of a positive outcome is then predicted by 
	\begin{align}
		\label{Eq:pred}
		P_i = \dfrac{\exp \big( \widehat{\beta}_0 + \widehat{\beta}_1 X_i \big)}  {1+\exp( \widehat{\beta}_0 + \widehat{\beta}_1 X_i )}.
	\end{align}  
	
	We vary the severity of the misspecification, expressed through the magnitude of $\beta_2$.
	Following \citet{Hosmeretal1997}, we characterize the ``lack of linearity'' through the conditions $\bar{\pi}(-3) = j - 0.00733745$, $\bar{\pi}(-1.5) = 0.05$ and $\bar{\pi}(3)= 0.95$ such that the value $j = 0$ results in the very accurate approximation $\beta_2 \approx 0$, i.e., a linear effect of $X_i$ on the log odds-ratio. 
	We consider a sequence of 51 equally spaced values of $j$ in the interval $[0, 0.1]$. 
	Notice that for each choice of $j$, the values of $\beta_0$ and $\beta_1$ are also determined by these conditions.

	\begin{table}[tb]
		\caption{Rejection rates in percentage points of the classical HL test and our eHL test under the null hypothesis with $j=0$ and the true regression parameters $(\beta_0, \beta_1)$ in \eqref{Eq:pred} at a significance level of $5\%$. 
			We treat an e-value above 20 as a rejection in the eHL test.}
	\label{tab:HLsize}
	\centering
	\begin{tabular}[t]{c c r c rrr} 
	\toprule
	&& HL && \multicolumn{3}{c}{eHL} \\ 
	\cmidrule(l{3pt}r{3pt}){3-3} 
	\cmidrule(l{3pt}r{3pt}){5-7}
	$s$ && && $1/3$ & $1/2$ & $2/3$ \\ 
	\midrule
	$n=1024$ && 6.2 && 0.5 & 1.0 & 0.4  \\ 
	$n=2048$ && 5.0 &&  0.1 & 0.4 & 0.6  \\ 
	$n=4096$ && 4.7 && 0.2 & 0.4 & 0.6  \\ 
	$n=8192$ && 4.5 && 0.0 & 0.1 & 0.5 \\ 
	\bottomrule
\end{tabular}
\end{table}

Table \ref{tab:HLsize} reports rejection rates of the tests over 1000 simulation replications, where we set $\beta_2 = 0$ (i.e., $j=0$), and use the true regression parameters $(\beta_0, \beta_1)$ in \eqref{Eq:pred} to guarantee that the null hypothesis $\mathcal{H}_{\text{HL},n}$ holds.
For the classical HL test, we use ten equally populated (quantile-spaced) bins, where the exact procedure follows the method \textsf{Q}{$^R$} described in Appendix \ref{sec:HLInstability}.
For the feasible eHL test of Section \ref{sec:feasible}, we use the splitting fractions $s \in \{1/3, 1/2, 2/3\}$.
To limit computation time, we choose a relatively low amount of bootstrap replications $B=10$ in the eHL test as we are mainly interested in rejection rates averaged across simulation replications, and hence, stability of the test is less of a concern as e.g., in the subsequent empirical application.
Here and in the following, we treat e-values above 20 as a test rejection at the $5\%$ significance level.
The table shows that all tests are well sized, where all eHL versions exhibit rejection frequencies much below the nominal value of $5\%$, which is not unusual for tests based on e-values.

\begin{figure}[tb]
\includegraphics[width=\linewidth]{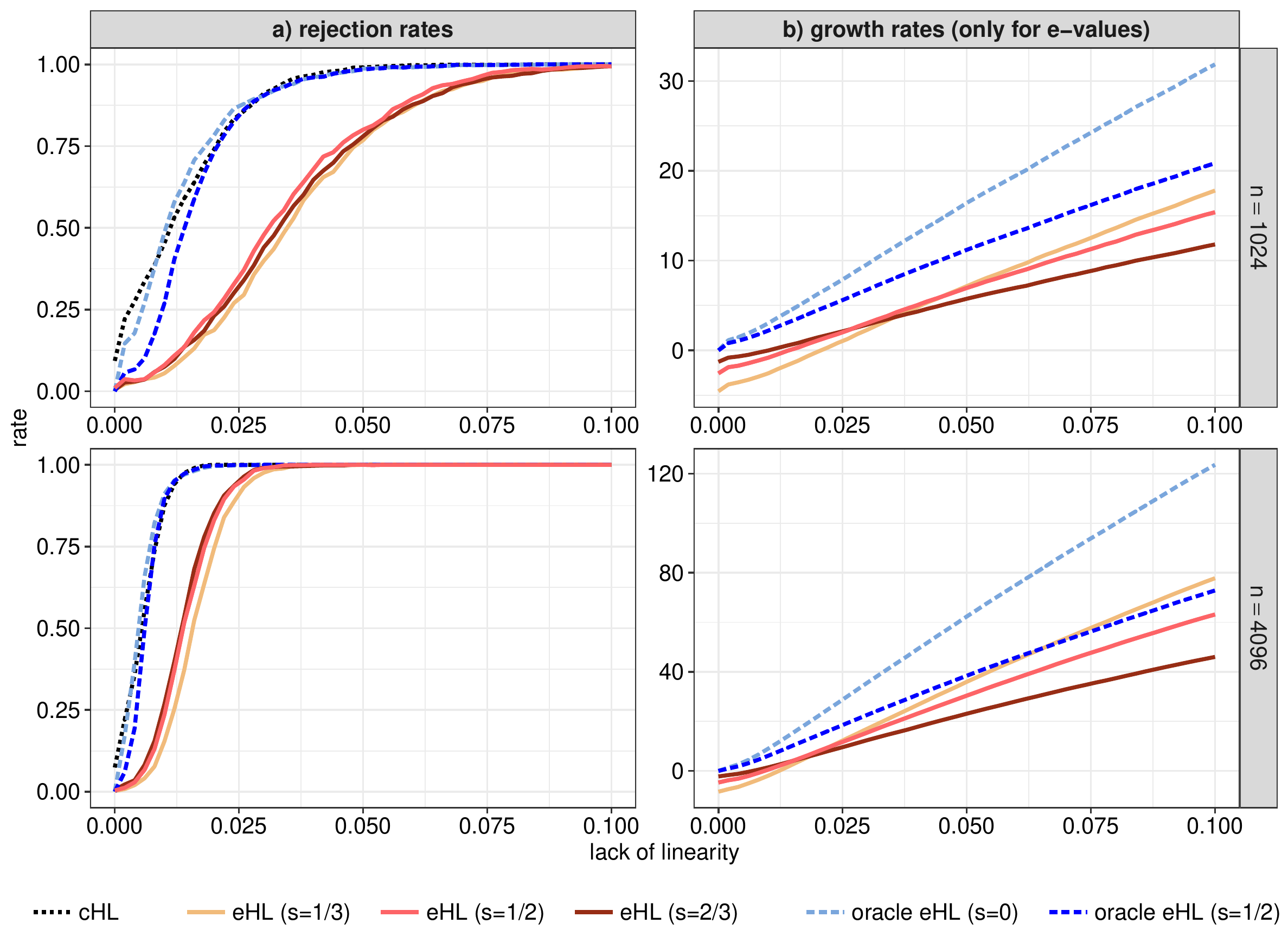}
\caption{Rejection (left) and e-value growth (right) rates for the classical HL (cHL) test, the feasible eHL test and an oracle eHL test for a range of splitting factors $s$.
	The oracle eHL test is based on the true $\pi_i$.
	The $x$-axis contains the severity of model mispecification, and the vertically aligned plots correspond to different sample sizes.} 
\label{fig:Rejection_Rates}
\end{figure}

Figure \ref{fig:Rejection_Rates} analyzes the tests' behaviour under the alternative hypotheses induced by $j>0$.
Notice that we use the true parameters $(\beta_0, \beta_1)$ in \eqref{Eq:pred} for $j=0$ but estimates $\big(\widehat{\beta}_0, \widehat{\beta}_1\big)$ for any $j>0$ as the pseudo-true parameters are unknown under model mispecification. In this analysis, we further include an oracle version of the eHL test, whose e-values are optimal in the sense that they are based on $q_i = \bar{\pi}_i$, i.e., the practically unknown true conditional event probabilities.
The oracle eHL version with $s=1/2$ facilitates a fair comparison with the feasible HL test based on the same splitting factor.
The left panel of the figure shows classical test rejection rates for a nominal significance level of $5\%$.
Following the explanations in Section \ref{sec:IdealTest} together with \citet{GrunwaldHeideETAL2019} and \citet{Shafer2021}, a suitable measure of power for e-values is the growth rate $\mathbb{E} \big[ \log(E_{\text{HL},n}) \big]$, which is shown in the right panel Figure \ref{fig:Rejection_Rates}, where we approximate the expectation by the average log e-value over the simulation replications.
We restrict attention to $n \in \{1024, 4096\}$ as the other sample sizes do not yield further insights.

We find that all tests develop power for increasing mispecification $j$. 
E.g., the feasible eHL versions already have substantial power for both sample sizes against an alternative with $j \approx 0.043$, which \citet{Hosmeretal1997} interpret as a value inducing only `slight' mispecification. (Notice that $j \approx 0.043$ equals 0.05 in their parametrization.)
There seems to be little difference among the feasible eHL tests when using different splitting fractions $s$, and hence, we do not find arguments to deviate from the natural choice of $s=1/2$, which we continue to use in the application.

The higher power of the classical HL test can be explained by the required sample split in the eHL test, and the estimation error in assigning suitable values for $q_i$.
The two oracle eHL tests make these steps redundant and hence achieve comparable power to the classical HL test.
Perhaps surprisingly, the difference between the two oracle eHL tests with different $s$ is smaller than the respective difference to the feasible test versions based on estimated $q_i$'s, which means that tuning the test to a specific alternative through the $q_i$'s is the main empirical challenge of the eHL test.

Notice that the often overlooked bin specification in the classical HL test implicitly determines the set of alternatives the test has power against as e.g.\ illustrated in \citet[Section 5]{Dimitriadis2022Honest}. As the sample split in the eHL test allows for \emph{estimating} a suitable alternative, Theorem \ref{thm:power} shows that the (ideal version of the) eHL test has power against all alternatives.

Turning to the growth rates of the feasible eHL tests, we find that larger choices of $s$ perform better for slight model mispecifications (small $j$) while the opposite is true for large mispecifications.
This can be explained since as discussed around \eqref{eq:alternative_q<P}--\eqref{eq:alternative_q>P}, $\bar{\pi}_i$ must be on the `correct' side of $P_i$ to gain power, which might be violated for small $s$ (and $n$) under slight mispecifications.


\section{Application: Credit Card Defaults in Taiwan}
\label{sec:application}


In this application, we analyze (re-)calibration of probability predictions for the binary event of credit card defaults in Taiwan in the year 2005.
In that time period, banks in Taiwan over-issued credit cards, also to unqualified applicants, who at the same time overused the cards for consumption, resulting in severe credit card debts  and damaged consumer finance confidence \citep{YehLien2009, LoHarvey2011}.
This crisis calls for improved and in particular calibrated probability predictions for credit card defaults that can be used for a thorough risk management and improved financial regulations.

For our analysis, we use a data set of $m = 30\,000$ credit card holders from Taiwan in 2005, that is publicly available from the UCI Machine Learning Repository \citep{UCIMLRep, YehLien2009} under \href{https://archive.ics.uci.edu/ml/datasets/default+of+credit+card+clients}{https://archive.ics.uci.edu/ml/datasets/default+of+credit+card+clients}.
Specifically, the binary response variable $Y_i \in \{0,1\}$ contains information on whether a default payment, $Y_i = 1$, occurred for customer $i=1,\dots,m$.
We observe a relatively high rate of $22.12\%$ of default payments in the data set that reflects the above mentioned credit card crisis.
The data set further includes 23 explanatory variables, containing information on the amount of given credit, gender, education, marital status, age, and various historical payment records for the past six months.

\begin{table}[tb]
\centering
\caption{E-values of the eHL and the range of p-values of the classical HL test, the latter stemming from 80 reasonable binning procedures as detailed in Table \ref{tab:HLinstabilities} and Appendix \ref{sec:HLInstability}.} 
\label{tab:application_results}
\begin{tabular}[t]{l l l}
	\toprule
	Prediction method & eHL e-values  & Range of HL p-values \\
	\midrule
	Logistic model & \hspace{0.5em}$7.0 \cdot 10^{28}$ & [0.00, 0.00] \\
	Logistic model with increased estimation set & \hspace{0.5em}$9.6 \cdot 10^{22}$ & [0.00, 0.00] \\
	Isotonic recalibration &  20.04 & [0.00, 0.91] \\
	Bagged isotonic recalibration &   \hspace{0.5em}6.14 & [0.00, 0.53] \\
	\bottomrule
\end{tabular}
\end{table} 

We randomly split the data into an estimation and a Recalibration set $\mathcal{R}$ with $M=12000$ observations each, and a Validation set $\mathcal{V}$ containing the remaining $n=6000$ observations.
We use the estimation set to fit a standard logistic regression model based on all predictor variables by maximum likelihood and compute the model predictions on the recalibration and validation sets, respectively.
We run all the following tests on the validation set.

Table \ref{tab:application_results} reports the e-values of the feasible version of our calibration test described in Section \ref{sec:feasible} based on $B=10000$ bootstrap replication and with a splitting factor of $s=1/2$ that is motivated by our simulation results.
We further report the range between the smallest and largest p-value of the classical HL test, where the different p-values result from five different, but natural binning procedures using $g=5,\dots,20$ bins, respectively.
We provide further details on these implementation choices in Appendix \ref{sec:HLInstability}.

The predictions from the logistic model result in an e-value far beyond the value of 20 in Table \ref{tab:application_results}, hence implying that these predictions are clearly miscalibrated.
In this setting, all implementation choices of the classical HL test agree and deliver p-values very close to zero. 
The second row of the table shows that even when using all observations in the ``increased estimation set'' comprising the estimation and the recalibration set, the situation barely changes and both the eHL and HL tests agree (under all implementation choices).

As a consequence of these clear rejections, we now aim at isotonically recalibrating the probability predictions, a technique that proved valuable in other disciplines \citep{Guo2017, Vannitsem2018}, where it is also called ``post-processing''.
For this, we estimate an isotonic regression on the recalibration set $\mathcal{R}$ and generate recalibrated predictions by transforming the logistic predictions on $\mathcal{V}$ with the estimated isotonic regression function.
Table \ref{tab:application_results} shows that our calibration test has an e-value just above 20, i.e., a weak rejection when interpreted as a (conservative) test at the $5\%$ level.

\begin{figure}[tb]
\centering
\includegraphics[width=\linewidth]{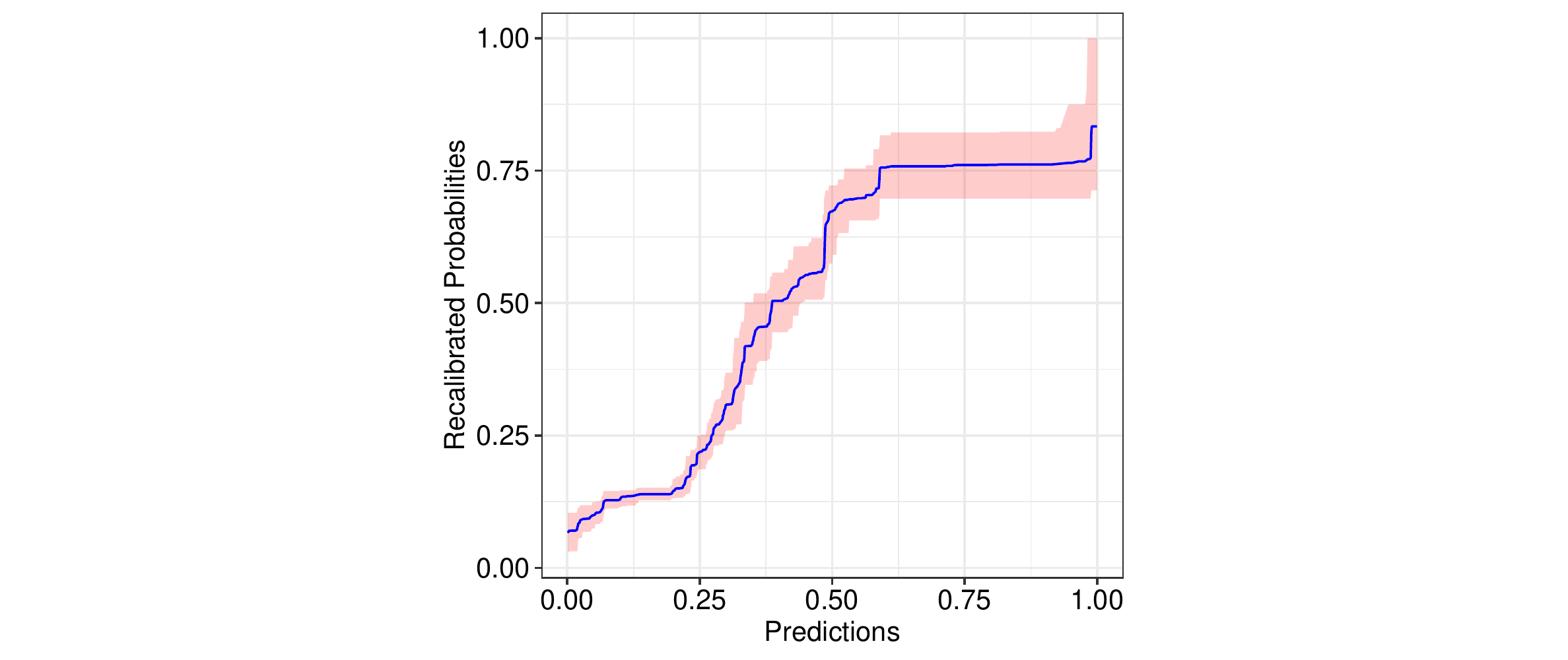}
\caption{Bagged isotonic recalibration curve of the logit predictions. 
	The blue curve shows the mean, and the red band the range of the pointwise $1\%$ and $99\%$ quantiles, over all bagging iterations.} 
\label{fig:recalibration_curve}
\end{figure}

As a nonparametric method, the isotonic regression is known to involve substantial estimation noise that might adversely affect the recalibrated predictions.
Hence, we stabilize the estimation through the classical bagging (bootstrap aggregation) method of \citet{Breiman1996bagging}.
In detail, we draw $\widetilde{B} = 100$ bootstrap samples $\mathcal{R}_b, b=1,\dots,\widetilde{B}$  of size $M$ from the recalibration set $\mathcal{R}$ and estimate the isotonic regression on each  bootstrap sample $\mathcal{R}_b$.
The final predictions are obtained by recalibrating with the average of the estimated isotonic regression functions, displayed in Figure \ref{fig:recalibration_curve}.

The last row of Table \ref{tab:application_results} shows an e-value of approximately 6 implying only very weak evidence against the null hypothesis of calibration, once again illustrating the practical strength of both, bagging and recalibration methods.
The estimated re-calibration function displayed in Figure \ref{fig:recalibration_curve} reinforces the importance of recalibrating the logistic model predictions by showing that it substantially deviates from the diagonal. 

\begin{figure}[tb]
\centering
\includegraphics[width=\linewidth]{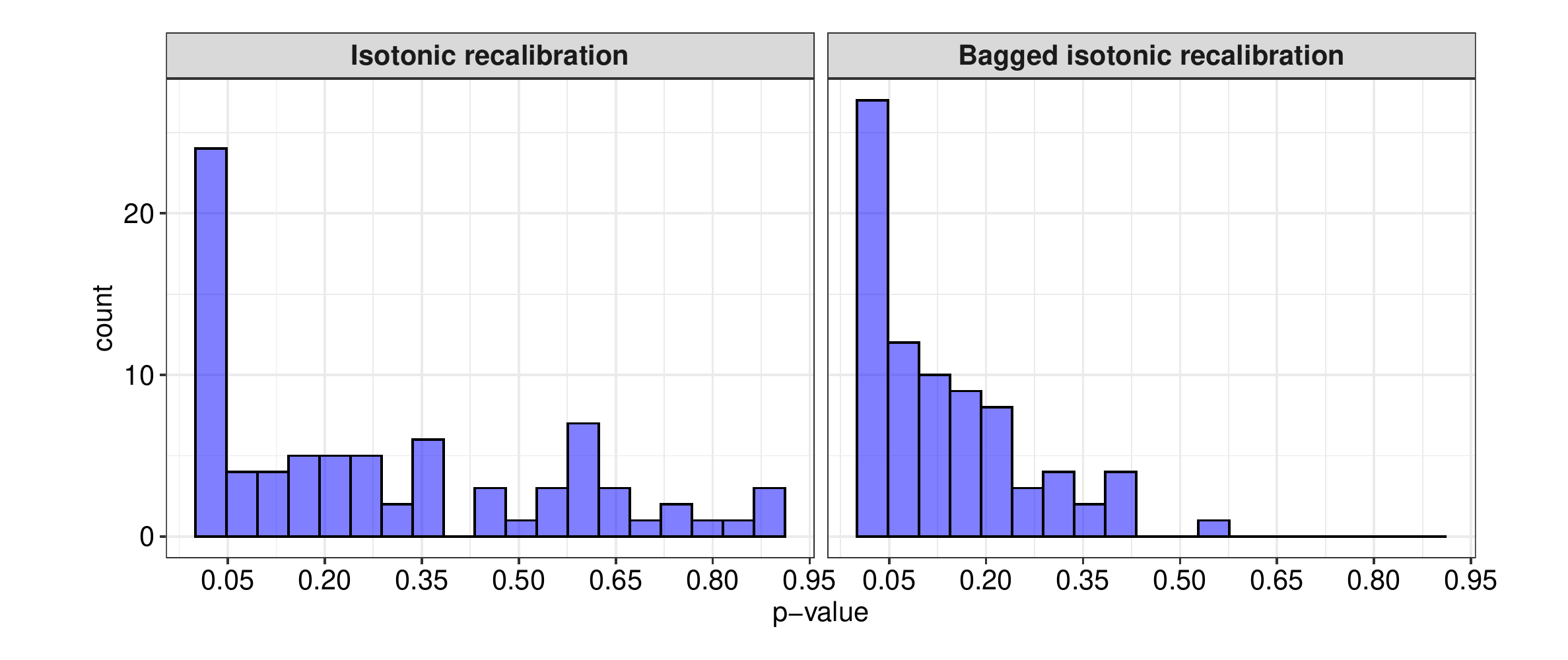}
\caption{Histograms of $p$-values of the classical out-of-sample HL test based the five binning procedures given in Appendix \ref{sec:HLInstability} based on 5-20 bins, respectively, resulting in a total of 80 test results.} 
\label{fig:HLinstabilities}
\end{figure}

For these two recalibration methods, the various natural implementation choices of the HL test, further described in Appendix \ref{sec:HLInstability}, result in p-values ranging between essentially 0 and 0.91 (and 0.53 respectively).
The corresponding p-value histograms in Figure \ref{fig:HLinstabilities} (and the detailed results in Table \ref{tab:HLinstabilities}) show the continuum of p-values, where the null hypothesis is rejected in approximately half of the cases at the $5\%$ level, implying that a researcher can essentially tailor the test decision to her will.
As already noted by \citet{Hosmeretal1997, Bertolini2000, Kuss2002}, this is a disconcerting state of affairs for a commonly used testing procedure and calls for more robust alternatives, such as the eHL test proposed in this paper.
Appendix \ref{sec:HLInstability} further shows that the feasible eHL test version is affected less by such instabilities arising from the repeated sample splits, at least if $B$ is chosen sufficiently large as in this application.

\section{Discussion \label{sec:discussion}}

This article proposes an e-test for perfect calibration, which is a safe testing counterpart to the widely used Hosmer-Lemeshow test. The proposed eHL test follows a simple betting interpretation (see \cite{Shafer2021}) where the e-value can be seen as the factor by which we multiply the bet against the hypothesis of perfect calibration. Intuitively, when accumulating money by the bet, we gain evidence against the null. Here, the e-value depends on the probability prediction, its corresponding realization, and an arbitrary value, which we suggest estimating in a two-step approach by isotonic regression. Furthermore, we assess the empirical performance of the test to detect quadratic model misspecifications. The simulations show that in samples of more than 2000 observations, the eHL test allows to reliably detect levels of quadratic misspecification, which \citet[p.~973]{Hosmeretal1997} denote to be slight. The intrinsic flexibility of the e-values allows the application of stable data-driven methods (here isotonic regression) instead of the typical binning and counting technique in the HL test. However, this flexibility comes at the cost of lower power in small samples of less than 2000 observations. 

The feasible version of our test is based on random splits of the training data. Since the null hypothesis $\mathcal{H}_{\text{HL},n}$ requires calibration conditional on the predictions $P_1, \dots, P_n$, one also obtains a valid tests if the splits are performed systematically based on $P_1, \dots, P_n$, for example, by choosing two subsets with similar distribution of the $P_i$ in order for the isotonic regression estimator to extrapolate well. Systematic sampling approaches to increase the power have already been applied by \citet{Duan2022} for testing treatment effects.

Our article focuses on the batch setting where a fixed sample of size $n$ is available, rather than the online setting in which $(P_i, Y_i)$, $i \in \mathbb{N}$, arrive sequentially. However, the fact that powerful tests based on isotonic regression can be constructed in the batch setting suggests that similar approaches may be fruitful for online testing. \citet[Section 7.1]{Kotlowski2016} describe algorithms with sublinear regret for online isotonic regression (without the random permutation setting). We believe that in conjunction with parts (i) and (ii) of our Theorem \ref{thm:power}, it is possible to derive power guarantees for sequential calibration tests where $\mathbb{E}[Y|P]$ is estimated sequentially with isotonic regression. We leave such extensions for future work.

\section*{Acknowledgments}

A.~Henzi and J.~Ziegel gratefully acknowledge financial support from the Swiss National Science Foundation.
T.~Dimitriadis gratefully acknowledges financial support from the German Research Foundation (DFG) through grant number 502572912.


\bibliographystyle{apalike}
\bibliography{biblio}




\vspace{1cm}
\appendix 
\noindent {\huge \textbf{Appendix}}

\section{(In-)Stability Results for the HL and eHL Tests}
\label{sec:HLInstability}

The classical HL test given in \eqref{eq:HLTest} is based on a partition of the unit interval into $g \in \mathbb{N}$ bins. 
We use the subsequently described five partitioning methods in the application in Section \ref{sec:application}, starting with the equidistant variant:
\begin{itemize}
\item   
$\textsf{E}$:
We partition the interval $\big[\min(p_i; i=1,\dots,n), \, \max(p_i; i=1,\dots,n) \big]$ into $g$ equidistant bins that are, apart from the first bin, open at left and closed at right.
\end{itemize}

We further use four natural implementations of ``quantile-based'' binning, all using a nominal number of $g$ bins.
These methods mainly differ for multiple occurrences of the same forecast value, which is however not unusual in practice and is e.g., an inherent feature of methods based on decision trees or isotonic regressions.

\begin{itemize}
\item 
$\textsf{Q}^L$: 
We partition the interval $[0,1]$ into $g$ left-open and right-closed bins according to the sample quantiles (using the default \texttt{quantile()} function in \texttt{R}) at levels $1/g,\dots,(g-1)/g$.
This method is denoted with the superscript $L$ as forecasts on the bin boundary are assigned to the Left bin.
The first bin is also closed at left and if the sample quantiles at different levels coincide, they are ignored, resulting in possibly less than $g$ bins.

\item 
$\textsf{Q}^R$: 
As $\textsf{Q}^L$, but we use $g$ right-open and left-closed bins such as forecasts on the bin boundaries are assigned to the Right bin.

\item 
$\textsf{Q}^+$:
We sort the forecast-realization pairs $(p_i, y_i)_{i=1}^n$ by their forecast values $p_i$ and in the case of tied forecast values, by their realizations in \emph{ascending} order.
Based on this order, we place the observations in $g$ equally populated bins.
If the size of the data set is not a multiple of $g$, excess values are redistributed in such a way that the bins with an additional observation are as far apart from each other as possible.

\item 
$\textsf{Q}^-$:
As variant $\textsf{Q}^+$, except that we sort in \textit{descending} order of $y_i$ for tied forecast values.
\end{itemize}

A comparison of the methods $\textsf{Q}^L$ and $\textsf{Q}^R$ illustrates that assigning predictions on the bin boundaries either to the left or right bins can have consequential implications.
The methods $\textsf{Q}^+$ and $\textsf{Q}^-$ circumvent this issue by selecting approximately equal amounts of observations into each bin, but in turn are sensitive to a change in the simple ordering of the (iid) observations in the underlying data, something that is usually ignored in applications.

\begin{table}[tb]
\centering
\caption{p-values of the HL test based on various binning choices described in the text for the two recalibrated prediction methods from Section \ref{sec:application}.}
\label{tab:HLinstabilities}
\begin{tabular}{r c ccccc c ccccc}
	\toprule
	\multicolumn{1}{c}{} && \multicolumn{5}{c}{Isotonic recalibration} && \multicolumn{5}{c}{Bagged isotonic recalibration} \\
	\cmidrule(l{3pt}r{3pt}){3-7} \cmidrule(l{3pt}r{3pt}){9-13}
	Bins && $\textsf{Q}^L$ & $\textsf{Q}^R$ & $\textsf{Q}^{+}$ & $\textsf{Q}^-$ & \textsf{E} &&$\textsf{Q}^L$ & $\textsf{Q}^R$ & $\textsf{Q}^{+}$ & $\textsf{Q}^-$ & \textsf{E} \\
	\midrule
	5 && 0.34 & 0.46 & 0.09 & 0.00 & 0.24 && 0.08 & 0.05 & 0.09 & 0.00 & 0.11\\
	6 && 0.16 & 0.60 & 0.22 & 0.00 & 0.31 && 0.10 & 0.21 & 0.18 & 0.26 & 0.33\\
	7 && 0.24 & 0.56 & 0.01 & 0.00 & 0.00 && 0.02 & 0.16 & 0.01 & 0.00 & 0.37\\
	8 && 0.59 & 0.55 & 0.17 & 0.00 & 0.38 && 0.11 & 0.20 & 0.17 & 0.02 & 0.15\\
	9 && 0.20 & 0.53 & 0.06 & 0.00 & 0.02 && 0.08 & 0.20 & 0.07 & 0.00 & 0.26\\
	10 && 0.26 & 0.67 & 0.19 & 0.00 & 0.36 && 0.20 & 0.11 & 0.22 & 0.00 & 0.10\\
	11 && 0.27 & 0.33 & 0.08 & 0.00 & 0.77 && 0.06 & 0.09 & 0.10 & 0.00 & 0.08\\
	12 && 0.15 & 0.91 & 0.19 & 0.00 & 0.02 && 0.10 & 0.18 & 0.21 & 0.01 & 0.11\\
	13 && 0.57 & 0.58 & 0.27 & 0.00 & 0.60 && 0.16 & 0.17 & 0.31 & 0.00 & 0.00\\
	14 && 0.22 & 0.87 & 0.01 & 0.00 & 0.09 && 0.03 & 0.07 & 0.01 & 0.00 & 0.03\\
	15 && 0.60 & 0.68 & 0.04 & 0.00 & 0.64 && 0.04 & 0.09 & 0.06 & 0.00 & 0.00\\
	16 && 0.80 & 0.28 & 0.17 & 0.00 & 0.11 && 0.29 & 0.37 & 0.20 & 0.00 & 0.01\\
	17 && 0.86 & 0.45 & 0.11 & 0.00 & 0.02 && 0.25 & 0.07 & 0.14 & 0.00 & 0.00\\
	18 && 0.36 & 0.63 & 0.14 & 0.00 & 0.10 && 0.19 & 0.19 & 0.17 & 0.00 & 0.00\\
	19 && 0.48 & 0.73 & 0.38 & 0.00 & 0.01 && 0.40 & 0.53 & 0.42 & 0.00 & 0.10\\
	20 && 0.59 & 0.83 & 0.35 & 0.00 & 0.61 && 0.42 & 0.30 & 0.39 & 0.00 & 0.02\\
	\bottomrule
\end{tabular}
\end{table}

While the existing literature often simply refers to ``quantile-based'' binning, this list shows that the HL test is sensitive to subtleties that one might easily disregard, but turn out to be consequential for the test result in some instances.
This is illustrated by Table \ref{tab:HLinstabilities}, which reports $p$-values for the classical HL test based on the five binning methods discussed above using $g=5,\dots,20$ bins respectively for the two recalibration methods used in the application in Section \ref{sec:application}.
We find that the p-values vary substantially in both, using different numbers of bins and different binning implementations. 
Maybe surprisingly, even for a fixed $g$, the subtleties in the four quantile-based binning choices lead to widely varying p-values.

\begin{figure}[tb]
\centering
\includegraphics[width=\linewidth]{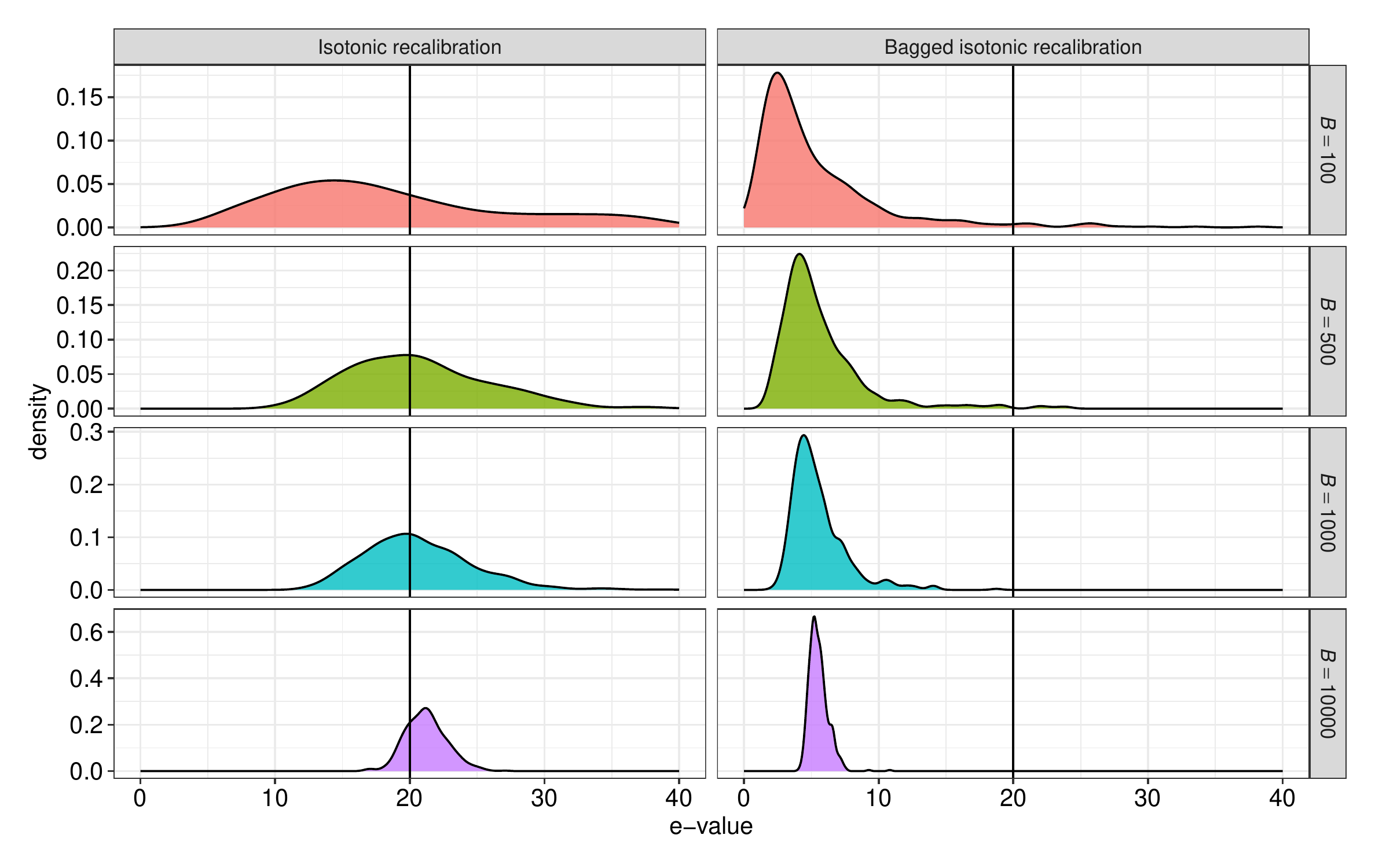}
\caption{Kernel density estimates of 500 e-values, obtained by starting the feasible test version from different random seeds and hence implying different random splits, for the two recalibrated prediction methods in the application based on $B = \{100,500,1000,10000\}$ bootstrap replications (in Algorithm \ref{alg:feasible}).
	For the bagged recalibration model, we observe 17 e-values above 20 for $B=100$, 5 for $B=500$, and none for $B \in \{1000, 10000\}$.} 
\label{fig:HLinstabilities}
\end{figure}

In contrast to the classical HL test, the theoretical version of the eHL test described in Algorithm \ref{alg:esym} is tuning parameter free due to the use of the isotonic regression method. 
This is unfortunately not true for the feasible eHL test described in Algorithm \ref{alg:feasible} that might be sensitive to the chosen sampling splits in the bootstrap-like replications.
In particular, one has to choose the number of replications $B$ large enough such that the resulting e-values are not sensitive to the random numbers (i.e., the `random seed') that determine the sample split.

To analyse this effect in our practical data example, Figure \ref{fig:HLinstabilities} visualizes the empirical distribution of the e-values (for tests based on different random splits), for varying bootstrap replications $B \in \{100, 500, 1000, 10000\}$.
While there is indeed some variation in the test result for smaller values of $B$, the e-values are relatively stable for $B=10000$, the choice we employ in the empirical application.
E.g., for the isotonic recalibration method, essentially all e-values are between 16 and 27, implying (conservative) p-values between $1/27 \approx 0.037$ and $1/16 = 0.0625$.
Similarly, the p-values in the bagged isotonic recalibration implied by the respective e-values range between $1/8 = 0.125$ and $1/4 = 0.25$.
In contrast, the variation of the HL test p-values in Table \ref{tab:HLinstabilities} is much more substantial and includes clear test rejections as well as many p-values above any commonly chosen significance level.

\end{document}